\documentclass[runningheads,a4paper]{llncs}

\usepackage[cmex10]{amsmath}
\usepackage{algorithmic}
\usepackage{tikz}
\usetikzlibrary{arrows,automata,backgrounds}

\usepackage{graphicx}
\usepackage{amssymb}
\usepackage{mathabx}
\usepackage{algorithm}
\usepackage{varwidth}
\usepackage{listings}
\usepackage{subfigure}
\usepackage{verbatim}
\usepackage{wrapfig}
\renewcommand{\qed}{\hspace{\stretch1}$\blacksquare$}

\begin{document}

\mainmatter  

\title{Probabilistic Model Checking of Incomplete Models}

\author{Shiraj Arora \and M. V. Panduranga Rao }

\institute{Indian Institute of Technology Hyderabad\\
India \\
\{cs14resch11010, mvp\}@ iith.ac.in}

\maketitle

\begin{abstract}
It is crucial for accurate model checking that the model be a complete and faithful representation of the system. Unfortunately, this is not always possible, mainly because of two reasons: (i) the model is still under development and (ii) the correctness of implementation of some modules is not established. In such circumstances, is it still possible to get correct answers for some model checking queries?

This paper is a step towards answering this question. We formulate this problem for the Discrete Time Markov Chains (DTMC) modeling formalism and the Probabilistic Computation Tree Logic (PCTL) query language. We then propose a simple solution by modifying DTMC and PCTL to accommodate three valued logic. The technique builds on existing model checking algorithms and tools, obviating the need for new ones to account for three valued logic. 

One of the most useful and popular techniques for modeling complex systems is through discrete event simulation. Discrete event simulators are 
essentially code in some programming language. We show an application of our approach on a piece of code that contains a module of unknown correctness.


\emph{A preliminary version of this paper appears in the proceedings of Leveraging Applications of Formal Methods, Verification and Validation:
Foundational Techniques (ISoLA 2016), LNCS 9952, Springer.}

\end{abstract}

\section{Introduction}
Probabilistic model checking is an important technique in the analysis of stochastic systems. Given a formal
description of the system in an appropriate modeling formalism and a requirement specification in an
appropriate system of formal logic, the problem is to decide whether the system
satisfies the requirement specification or not. Popular model checking techniques for such systems include numerical 
model checking which is expensive but accurate, and statistical model checking wherein accuracy can be traded off for
speed~\cite{courcoubetis88,RV10a,NumVsStat}.

Modeling formalisms for stochastic systems are usually variants of Markov Chains like Discrete and Continuous Time 
Markov Chains (DTMC and CTMC respectively)~\cite{CTMC}, 
Constrained Markov Chains~\cite{CMC} and Probabilistic 
Automata~\cite{APA}.
Specification requirement queries are typically formulated in logics like 
Probabilistic Computation Tree Logic (PCTL)~\cite{HanssonJ94} and Continuous Stochastic 
Language (CSL)~\cite{ModelCheck}.

However, for more complex systems, it is convenient to use more powerful modeling techniques like
Discrete Event Simulation (DES) and agent based simulation, and statistical model checking for analysis~\cite{Younes}. 
Indeed, statistical model checkers that can be coupled with discrete event simulators have been designed. 
Tools like PLASMA~\cite{PLASMA1,PLASMA2} and
MultiVesta~\cite{MultiVesta}, 
which builds on the statistical model checker Vesta~\cite{Vesta} and its parallel variant PVesta~\cite{PVesta} are recent
popular examples.

While substantial
work has been done in the model checking domain, 
important practical problems can arise due to the quality of the simulation tool itself.
For example, there could be stubs for unwritten modules in the simulation tool, or modules whose correctness is not yet established. It is not
clear how good such a simulator is for the purpose of model checking. Is it, for example, impossible to
verify the satisfaction of a given query on such an implementation? Or is it the case that in spite of lacunas
in the implementation, some model checking queries can still be answered?

In this paper, we demonstrate a simple algorithm towards answering this question. 
The central idea originates from the observation that at an abstract level, the problem boils down
to the inability of assigning truth values to atomic propositions in a state of the model. 
We demonstrate the approach using appropriately modified DTMC and PCTL. The proposed modifications are as
follows: In the state of a DTMC, an atomic proposition can take the value $Unknown$
(abbreviated ``$?$") in addition to the usual $True$ ($T$) or $False$ ($F$). The syntax and semantics of PCTL are modified so 
that a PCTL formula can  also take the value ``$?$".

Intuitively, the question that we ask is: Are there a sufficient number of paths in the DTMC that do not evaluate 
to ``$?$"? If so, does the modified PCTL query evaluate to $True$ or to $False$ on this DTMC?
Our algorithm answers these questions by invoking the model checking tool twice (PRISM~\cite{prism} in our case) as
a subroutine. This is a crucial advantage, as it means that the model checker itself need not be changed to account for three valued logic.

We illustrate applications of the algorithm with examples of varying complexity. In particular, we demonstrate the 
usefulness of our approach with an example program that has a module of unknown correctness.

The paper is arranged as follows. The next section briefly discusses some preliminary notations and definitions, and
relevant previous work done on model checking using three valued logic. Section 3 discusses our modifications 
in the definitions of DTMC and PCTL and the modified model checking algorithm. Section 4 discusses implementation details, the examples,
and the results. Section 5 concludes the paper with a brief discussion on future directions.

\section{Preliminaries and Related Work}
This section briefly discusses some basic definitions and terminology that will be used subsequently in the paper.
For details, see~\cite{ModelCheck}.

\subsection{Discrete Time Markov Chains (DTMC)}

A Discrete Time Markov Chain (DTMC) is one in which transition from one state to another occurs in 
discrete time steps. 

\begin{definition} A DTMC is a tuple $M = (S, \mathbb{P}, s_{init} , AP, L) $ where $S$ is a nonempty 
set of states, $\mathbb{P}$ : $S \times S \rightarrow [0, 1] $ is the transition probability function such that for all states 
$s\ \in\ S$ :\[\sum_{s'\ \in \ S}\ \mathbb{P}(s,s') = 1\] $s_{init} \in S$ is the initial state, $AP$ is a set of atomic
 propositions, and $L : S \rightarrow 2^{AP}$ is a labeling function, which assigns to each state a subset of $AP$ 
 that are true in that state.
\end{definition}

\begin{definition}
A \textbf{path} $\pi$ in a DTMC $M$ is a sequence of states $s_{0} ,s_{1}, s_{2}...$ such that for all $i=0,\ 1,\ 2,...$ , 
$\mathbb{P}(s_{i},s_{i+1})\ >0$.  The $(i+1)^{th}$ state in a path $\pi$ is written as $\pi[i]$. $Path(s)$ denotes the 
set of all infinite paths which start from a state $s$ in the model, $M$. $Paths_{fin}(s)$ is the set of all
 finite paths starting from state $s$.
\end{definition}

\begin{definition}
A \textbf{cylinder set}, $C(\omega)$ is the set of infinite paths that have a common finite prefix $\omega$ of length $n$. Let $\Sigma_{Path(s)}$
be the smallest $\sigma$-algebra generated by $\{C(\omega)\ |\ \omega \in Paths_{fin}(s)\}$. Then, we can define $\mu$ on the 
measurable space $(Path(s)$,$\Sigma_{Path(s)})$ as the unique probability measure such that:
\[\mu(C(\omega)) = \prod_{i=0}^{n-1} \mathbb{P}(s_{i},s_{i+1})\]
\end{definition}

\subsection{Probabilistic Computation Tree Logic (PCTL)}
Probabilistic Computation Tree Logic (PCTL), an extension of Computation Tree Logic (CTL), was introduced by Hansson and 
Johnson \cite{HanssonJ94} for analyzing discrete time probabilistic systems.\\
\\
\textbf{Syntax of PCTL:}
\[ \Phi\ ::=\  T\ |\ a\ |\ \Phi_{1} \land \Phi_{2} \ |\ \lnot\Phi \ |\ \mathbb{P}_{\bowtie \theta} [\psi]\]
\[ \psi\ ::=\ X \Phi \ | \ \Phi_{1} U\ \Phi_{2} \ | \ \Phi_{1} U^{\leq k}\  \Phi_{2}\]
where $\Phi$, $\Phi_{1}$, and $\Phi_{2}$ are state formulas, $\psi$ is a path formula, \textit{a} is an atomic proposition, 
$\theta \in [0,1]$ is the probability constraint, $\bowtie\ \in \{ <, \ >,\ \leq,\ \geq \}$ represents the set of operators, 
and $k\ \in\ \mathbb{N}\ $ is the time bound. The $ X$, $U$, and  $U^{\leq k} $ operators are called \textit{Next}, \textit{Until} 
and \textit{Bounded Until} respectively. \\
\\
\textbf{Semantics of PCTL:}

Let $M : (S,\mathbb{P}, s_{init},AP,L) $ be a Discrete Time Markov Chain. Let 
$s \ \in\ S$, $a\ \in\ AP$, and $\Phi,\ \Phi_{1},\ \Phi_{2}$ be PCTL state formulas, and $\psi$ be a PCTL path formula. 
Then, $\Phi$ is said to be satisfied in state $s$ i.e. $(s,\Phi)=T$ if:
\begin{center}
\begin{tabular}{ccc}
      $ (s,T)\ =\ T $, \\
      $ (s,a)\ =\ T $ & iff & $a\ \in L(s)$ , \\
      $ (s,\lnot\Phi)\ =\ T $ & iff & $(s,\Phi) = F $ , \\
      $ (s,\Phi_{1}\land\Phi_{2})\ =\ T $ & iff & $ (s,\Phi_{1})=T\ \land (s,\Phi_{2})=T $ ,\\
      $ (s,\mathbb{P}_{\bowtie\ \theta} (\psi))\ =\ T $ & iff & $\ \ \mu(\{\pi \in Path(s)\ |\ (\pi,\psi)\ =\ T \})\ \bowtie \theta$ 
\end{tabular}
\end{center}

If $\pi\ = s_{0}\ s_{1}\ s_{2}...\ $ is a path in $Path(s_{0})$ then \(\Pr( (s,\psi) =T)\ =\ \mu\{\ \pi \in Path(s)\ |\ (\pi,\psi)=T\}\) 
i.e. the probability of the set of paths starting from $s$ which satisfy the path formula $\psi$. The last satisfaction relation 
for a state formula thus states that the probability that $\psi$ is true on paths starting at $s$ satisfies
$\bowtie \theta$. A path formula $\psi$ is said to be satisfied for path $\pi$ i.e. $(\pi,\psi)=T$ if:
\begin{center}
\begin{tabular}{ccc}
$ (\pi, X \Phi)=T$ & iff & $(\pi[1],\Phi)=T $ , \\
$ (\pi, (\Phi_{1}\ U\ \Phi_{2}))=T$ & iff & [$\exists i \geq 0 \ |\ (\pi[i],\Phi_{2})=T ]\ \land [\ \forall j < i, (\pi[j],\Phi_{1})=T\ ]$ , \\
$ (\pi,(\Phi_{1}\ U^{\leq k}\ \Phi_{2} ))=T $ & iff & $\ [\exists\  i \leq k \ |\ (\pi[i],\Phi_{2})=T\ ]\ \land [\ \forall j < i, (\pi[j],\Phi_{1})=T\ ]. $
\end{tabular}
\end{center}

\noindent \textbf{Problem Statement for PCTL model checking:} Given a DTMC $M$, decide whether a PCTL formula $\Phi$
evaluates to $T$ or $F$ on $M$.

\subsection{Three-valued logic and model checking}

Multi-valued logics have been comprehensively investigated in the past few decades. In addition to having a rich
theory, they have also found practical applications. Depending on the problem, classical binary language can be extended
to include additional truth values.
For example, an additional truth value can be used to represent inconsistent and incomplete information. 
An application might also demand that we use two different values to denote inconsistent and incomplete information separately. 

In this work, we will use three valued logic. We expand the logic associated with atomic propositions in the state of a DTMC to 
include $Unknown$, denoted by the question mark symbol ``$?$". In what follows, we will use $Unknown$ and $?$ interchangeably. A number 
of different truth tables have been 
designed for three valued logics
 \cite{mal93many,putnam57three,manyval}. The three valued logic used in this work has all the properties of a Quasi-Boolean 
lattice and the 
truth tables for logic operations are described in Tables \ref{and}, \ref{or} and \ref{not}.

\begin{table}[h]
\centering
\begin{minipage}{0.3\textwidth}
\centering

\begin{tabular}{c| c c c}
        $\land$ & T & ? & F \\
        \hline
         T & T & ? & F \\ 
         ? & ? & ? & F \\
         F & F & F & F 
\end{tabular}
\caption{AND operator }
\label{and} 
\end{minipage}%
\hfill
\begin{minipage}{0.3\textwidth}
\centering
    \begin{tabular}{c| c c c}
        $\lor$ & T & ? & F \\
        \hline
         T & T & T & T \\ 
         ? & T & ? & ? \\
         F & T & ? & F 
    \end{tabular}
\caption{OR operator }
\label{or}
\end{minipage}%
\hfill
\begin{minipage}{0.3\textwidth}
\centering
    \begin{tabular}{c| c}
        $\lnot$ &  \\
        \hline
         T & F \\ 
         ? & ? \\
         F & T 
    \end{tabular}
\caption{NOT operator}
\label{not}
\end{minipage}

\end{table}

Indeed, three valued logic has been used in the past for model checking in non-probabilistic settings--for example, 
LTL~\cite{bruns99model,bruns00gen,genLTL} and CTL~\cite{chechik00,chechik01model}. Chechik et al.~\cite{chechik00,chechik01model}
have used three valued 
logic for atomic propositions as well as for the transition functions. In case of transition functions, the \textit{True} 
and \textit{False} values denote the presence or absence of a transition between two states respectively.
The third truth value represents the lack of information about the transition.

Three valued logic has also been used with numerical model checking of probabilistic systems, but with a different
motivation and solution. To overcome 
the problem of state-space explosion in numerical model checking, two or more states of a model are combined, yielding
an \emph{abstract} Markov chain.
However, over-abstraction often leads to a significant loss of information. Three valued logics have been 
  associated with abstract probabilistic systems wherein an $Unknown$ value represents loss of information, indicating
  that the level of abstraction should be decreased. Model checking of an abstract Markov chain 
is often done by reducing 
it to a Markov decision process and then using model checking techniques 
for Markov decision processes. For more details, please see~\cite{dontknow,huththree,klinkthree}. 

\section{Problem Statement and Solution}

As mentioned earlier, the aim of this work is to study the effect of an information being unknown, 
in asserting whether a given property is satisfied in the model 
or not. To perform model checking on such three valued systems, both DTMC and PCTL 
need to be modified. While in case of DTMC the labeling function $L$ is modified,
the semantics are altered for PCTL. 
Intuitively, in our approach, the model checker aims to identify if there are too many paths 
in a model wherein it is not known whether a property will be satisfied or not. Thus, 
the model checker first evaluates whether the property is satisfied in the model and 
if not, it examines the reason behind the lack of satisfaction. 

In the coming subsections, we discuss the modifications in DTMC and PCTL, and the problem statement.

\subsection{DTMC with Question Marks}
A Discrete Time Markov Chain with question marks (qDTMC) is a tuple $M : (S,\mathbb{P},s_{init},AP,L)$ with a finite 
non-empty set of states $S$, a transition probability function $\mathbb{P}:S\times S \rightarrow [0,1]$ such 
that for all states $s\ \in\ S$ : \(\sum_{s'\ \in \ S}\ \mathbb{P}(s,s') = 1\), the initial state $s_{init} \in S$, 
a set of atomic propositions $AP$ and labeling function $L:S\times AP\rightarrow \{T,F,\textrm{{\bf ?}}\}$. 

\subsection{PCTL with Question Marks}

The syntax of PCTL in the context of three valued logic (hereafter referred to as qPCTL for convenience) remains
the same. The operators ($\land$, $\lor$, $\lnot$) and operands ($T,F,?$) however, are as defined in 
Tables 1, 2 and 3 for three valued logic. Therefore, the structure of the queries remains unchanged. 
However, the semantics need to be modified:

\noindent \textbf{Semantics:}

Let $M : (S,\mathbb{P},s_{init},AP,L) $ be a qDTMC model. 
Let $s \ \in\ S$, $a\ \in\ AP$, $\Phi$, $\Phi_{1}$, $\Phi_{2}$ be qPCTL state formulas, and $\psi$ be a qPCTL 
path formula. Then, semantics for  $\Phi$ are as stated below:
\begin{center}
       $(s,T)\ =\ T$ , \\
       $(s,F)\ =\ F$ , \\
       $(s,?)\ =\ ?$ . 
\end{center}
\begin{equation*}
  (s, a) = \left\{
      \begin{array}{ccc}
            T & \text{iff} & L(s,a)=T ,\\
	        F & \text{iff} & L(s,a)=F , \\
		    ? & \text{iff} & L(s,a)=? .
      \end{array} \right.
\end{equation*}

\begin{equation*}
 (s, \lnot\Phi) = \left\{
      \begin{array}{ccc}
            T & \text{iff} & (s,\Phi)=F,\\
	        F & \text{iff} & (s,\Phi)=T,\\
		    ? & \text{iff} & (s,\Phi)=? .
      \end{array} \right.
\end{equation*}
 
\begin{equation*}
  (s, \Phi_{1}\land\Phi_{2}) = \left\{
      \begin{array}{ccc}
            T & \text{iff} &  (s,\Phi_{1})=T\ \land (s,\Phi_{2})=T ,\\
	        F & \text{iff} &  (s,\Phi_{1})=F\ \lor (s,\Phi_{2})= F ,\\
		    ? & & otherwise.
      \end{array} \right.
\end{equation*}

The intuition behind the above definitions follows directly from three valued logic. 
The semantics of the probabilistic state formula are defined as follows:
\begin{equation*}
  (s, Pr_{\geq\theta}(\psi)) = \left\{
      \begin{array}{rl}
            T & \ \ \text{if } \mu\{\pi\in Path(s):(\pi,\psi)=T\} \geq \theta,\\
	        F & \ \ \text{if } \mu\{\pi\in Path(s):(\pi,\psi)=F\} \geq 1-\theta,\\
		    ? & \ \ \text{if } (\mu\{\pi\in Path(s):(\pi,\psi)=T\} < \theta) \\
		      & \ \ \land (\mu\{\pi\in Path(s):(\pi,\psi)=F\} < 1-\theta) .
			    \end{array} \right.
\end{equation*}

We first note that every formula must evaluate to one of $T,F,$ or $?$.
The above definition follows from the intuition that $(s, Pr_{\geq\theta}(\psi))$ evaluates to $T$ if 
at least $\theta$ fraction of the paths evaluate $\psi$ to $T$. Also, if $1-\theta$ (or more) fraction 
of the paths evaluate to $F$, then there are sufficient number of paths to evaluate $(s, Pr_{\geq\theta}(\psi))$ to $F$. 
However, if there does not exist enough paths to decisively tell whether or not the property $\psi$ holds, then 
$(s, Pr_{\geq\theta}(\psi))$ evaluates to $?$. We now turn to the semantics of the path formulas.

\begin{equation*}
  (\pi, X\Phi) = \left\{
      \begin{array}{rl}
            T & \text{if } (\pi[1],\Phi)=T ,\\
             F & \text{if } (\pi[1],\Phi)=F ,\\
		       ? & \text{if } (\pi[1],\Phi)=? .
			    \end{array} \right.
\end{equation*}

\begin{equation*}
  (\pi, \Phi_1 U^{\leq k} \Phi_2) = \left\{
      \begin{array}{rl}
            T & \text{if } \exists i\leq k: (\pi[i],\Phi_2)=T \wedge \forall i'< i: (\pi[i'],\Phi_1)=T\\
	          F & \text{if } (\forall i\leq k:( \pi[i],\Phi_2)=F) \vee (\exists i\leq k : (\pi[i],\Phi_2)\neq F \wedge
		  \\ &\qquad \qquad\qquad\qquad\qquad\qquad\ \ \  \exists i'< i :(\pi[i'],\Phi_1)=F), \\
		        ? & \text{otherwise}.
			    \end{array} \right.
\end{equation*}
			    
\begin{equation*}
  (\pi, \Phi_1 U \Phi_2) = \left\{
      \begin{array}{rl}
            T & \text{if } \exists i: (\pi[i],\Phi_2)=T \wedge \forall i'< i: (\pi[i'],\Phi_1)=T\\
	          F & \text{if } (\forall i:( \pi[i],\Phi_2)=F) \vee (\exists i : (\pi[i],\Phi_2)\neq F \wedge
		  \\ &\qquad \qquad\qquad\qquad\qquad \exists i'< i :(\pi[i'],\Phi_1)=F), \\
		        ? & \text{otherwise}.
			    \end{array} \right.
\end{equation*}

First, we note that $(\pi,X\Phi)$ evaluates to $T,F$ or $?$ depending on whether $\Phi$ is $T,F$ or $?$ in
$\pi[1]$ that is, in the next state. 

The \textit{bounded until} formulas in qPCTL are simple extension of 
the corresponding formulas in the standard PTCL. For example, a \textit{bounded until} formula 
$\Phi_1U^{\leq k}\Phi_2$ evaluates to $?$ if one of the following happens:
\begin{itemize}
\item $\Phi_2$ is $?$ in all the states up to $k$.
\item $\Phi_2$ is $?$ in at least one state and is never $T$ in any of the states along the path up to $k$, 
			but $\Phi_1$ is never $F$. 
\item $\Phi_2$ is $T$ for some $i\leq k$ and $\Phi_1$ is $?$ in at least one state but never $F$ in any of  
			the states upto $i$.
\end{itemize}
\noindent \textit{Unbounded until} is also extended similarly for qPCTL. We are now in a 
position to formally define the problem. \\

\noindent \textbf{Problem Statement:} Given a  qDTMC $M$, and a qPCTL formula $\Phi$, decide whether $(s_{init},\Phi)$
evaluates to $T$, $F$ or $?$.

\subsection{The Algorithm}

As mentioned earlier, our algorithm for model checking qDTMC using qPCTL uses binary model 
checkers as a subroutine. The algorithm involves modifying the input qDTMC suitably
before subjecting it to binary model checking queries. The central idea behind the algorithm is to use
these modifications to filter the three truth values successively, using only binary valued model checkers.
In what follows, for ease of exposition, we denote the outcomes of the binary model checker by $T'$ and 
$F'$. Algorithm \ref{qMC} describes our approach.
\begin{algorithm}
\caption{\textbf{qMC}}
\label{qMC}
\begin{algorithmic}
\vspace{2mm}
\STATE INPUT: A qDTMC $M$ and a qPCTL formula $\Phi$.
\IF[\emph{Conditional 1}]{$\Phi$ contains an AP of form $\lnot a$}
    \STATE Add new AP $a'= \lnot a$ in $M$.
    \STATE Replace all instances of $\lnot a$ with $a'$ in $\Phi$.
\ENDIF 
\STATE Set $?$ to $F$ in qDTMC $M$ to obtain binary DTMC $M^{(1)}$.
\IF[\emph{Conditional 2}]{BINARY\_MC$(M^{(1)}, \Phi) =T'$}
    \RETURN $T$
    \ELSE   
 \STATE  Set $?$ to $T$ in qDTMC $M$ to obtain binary DTMC $M^{(2)}$.
        \IF[\emph{Conditional 3}] {BINARY\_MC$(M^{(2)},\Phi) = F'$}
	       \RETURN $F$ 
	           \ELSE  
		           \RETURN $?$ 
		\ENDIF  
\ENDIF
\end{algorithmic}
\end{algorithm}

Intuitively, the algorithm proceeds as follows.
In each phase of the algorithm, $?$ truth values in the qDTMC $M$ are identified with either $T$ or $F$. 
Let $\Phi$ be a qPCTL formula consisting of an atomic proposition $a$. Then, all instances of the atomic 
proposition $a$ with truth value $?$ are set to $F$ in the first phase of the algorithm. On the other hand, 
the $?$ truth values are set to $T$ in the second phase of algorithm. While the first phase determines if 
sufficient paths with $a=T$ exist in the model to verify $\Phi$, the second phase checks for the paths 
with $a=F$ to disprove $\Phi$. 

For atomic propositions of the form $\lnot a$, 
the algorithm should search for paths with $a=F$ and $a=T$ in first and second phases respectively. Thus, 
the sequence of mapping needs to be reversed for atomic propositions of the form $\lnot a$. To maintain 
uniformity in the algorithm, for each atomic proposition of the form $\lnot a$ in $\Phi$, a new atomic 
proposition $a' = \lnot a$ is added in $M$. Similarly, each instance of $\lnot a$  is replaced with $a'$ 
in $\Phi$. Note that new atomic propositions need to be added only 
when a negated atomic proposition exists in $\Phi$. 

Figures \ref{example_old} and \ref{example_new} show a qDTMC $M_{1}$ 
in which a new atomic proposition $t = \lnot p$ is added. The modified models in the two phases of the Algorithm 
\ref{qMC} are given in figures \ref{small_step1} and \ref{small_step2}. 

\begin{figure}
\centering
    \begin{subfigure}[Example qDTMC $M_1$]{%
    \label{example_old}
        \begin{tikzpicture}[->,>=stealth',shorten >=1pt,auto,node distance=2.05cm, thin, state/.style={circle,draw,font=\sffamily\tiny,initial text=,minimum size=0.75cm, inner sep =0mm}]
    
    \node[initial,state,label=270 : {\tiny $s_{init}$}] (0)              {$\lnot p q?$};
    \node[state]         (2) [right of=0] {$\ p\ q  $};
    \node[state]         (1) [above of=2] {$p? \lnot q$};
    \node[state]         (3) [below of=2] {$\lnot p q?$};
    \node[state]         (4) [right of=2] {$p?q?$};
    \node[state]         (5) [right of=3] {$\ p \lnot q$};
    \node[state]         (6) [right of=1] {$\ p\ q  $};
  
    \path [every node/.style={font=\sffamily\tiny}]  
            (0) edge [bend left]    node {$0.3$}  (1)
                edge                node {$0.2$}  (2)
                edge                node {$0.5$}  (3)            
            (1) edge                node {$0.1$}  (0)
                edge [bend left]    node {$0.35$} (2)
                edge                node {$0.4$}  (4)
                edge                node {$0.25$} (6)                
            (2) edge                node {$0.1$}  (1)
                edge [bend left]    node {$0.1$}  (3)
                edge                node {$0.8$}  (4)
            (3) edge                node {$0.5$}  (2)
                edge [bend right]   node {$0.5$}  (4)
            (4) edge                node {$0.33$} (6)
                edge                node {$0.67$} (5)
            (5) edge [loop below]   node {$0.9$}  (5)
                edge                node {$0.1$}  (3)
            (6) edge [loop above]   node { $1$ }  (6);
\end{tikzpicture}
}
\end{subfigure}
\begin{subfigure}[Example qDTMC $M_1$ with new atomic proposition $t$, after Conditional 1 in Algorithm qMC]{
        \begin{tikzpicture}[->,>=stealth',shorten >=1pt,auto,node distance=2.05cm, thin, state/.style={circle,draw,font=\sffamily\tiny,initial text=,minimum size=0.75cm, inner sep =0mm}]
        \label{example_new}
    \node[initial,state,label=270 : {\tiny $s_{init}$}] (0)              {$\lnot p q?t$};
    \node[state]         (2) [right of=0] {$\ p\ q \lnot t $};
    \node[state]         (1) [above of=2] {$p? \lnot qt?$};
    \node[state]         (3) [below of=2] {$\lnot p q?t$};
    \node[state]         (4) [right of=2] {$p?q?t?$};
    \node[state]         (5) [right of=3] {$\ p \lnot q\lnot t$};
    \node[state]         (6) [right of=1] {$\ p\ q \lnot t\ $};
  
    \path [every node/.style={font=\sffamily\tiny}]  
            (0) edge [bend left]    node {$0.3$}  (1)
                edge                node {$0.2$}  (2)
                edge                node {$0.5$}  (3)            
            (1) edge                node {$0.1$}  (0)
                edge [bend left]    node {$0.35$} (2)
                edge                node {$0.4$}  (4)
                edge                node {$0.25$} (6)                
            (2) edge                node {$0.1$}  (1)
                edge [bend left]    node {$0.1$}  (3)
                edge                node {$0.8$}  (4)
            (3) edge                node {$0.5$}  (2)
                edge [bend right]   node {$0.5$}  (4)
            (4) edge                node {$0.33$} (6)
                edge                node {$0.67$} (5)
            (5) edge [loop below]   node {$0.9$}  (5)
                edge                node {$0.1$}  (3)
            (6) edge [loop above]   node { $1$ }  (6);
\end{tikzpicture}
}\end{subfigure}

\caption{First step of Algorithm qMC} \label{example}
\end{figure}
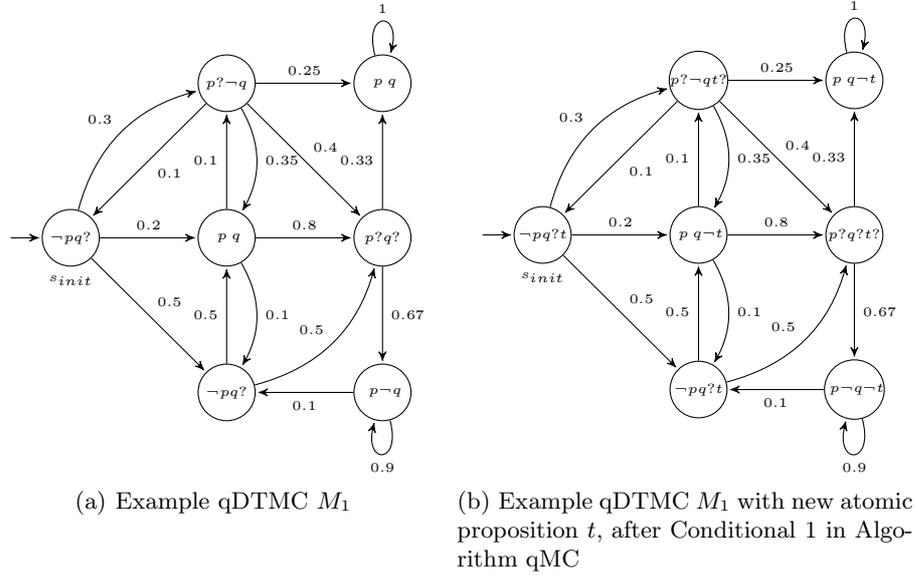

\begin{figure}
\centering
    \begin{subfigure}[Modified model $M_1^{(1)}$ in step 1.]{%
    \label{small_step1}
        \begin{tikzpicture}[->,>=stealth',shorten >=1pt,auto,node distance=2.05cm, thin, state/.style={circle,draw,font=\sffamily\tiny,initial text=,minimum size=0.75cm, inner sep =0mm}]
    
    \node[initial,state,label=270 : {\tiny $s_{init}$}] (0)              {$\lnot p \lnot qt$};
    \node[state]         (2) [right of=0] {$\ p\ q \lnot t $};
    \node[state]         (1) [above of=2] {$\lnot p \lnot q\lnot t$};
    \node[state]         (3) [below of=2] {$\lnot p \lnot q\ t$};
    \node[state]         (4) [right of=2] {$\lnot p \lnot q \lnot t$}; 
    \node[state]         (5) [right of=3] {$\ p \lnot q\lnot t$};
    \node[state]         (6) [right of=1] {$\ p\ q \lnot t\ $};
  
    \path [every node/.style={font=\sffamily\tiny}]  
            (0) edge [bend left]    node {$0.3$}  (1)
                edge                node {$0.2$}  (2)
                edge                node {$0.5$}  (3)            
            (1) edge                node {$0.1$}  (0)
                edge [bend left]    node {$0.35$} (2)
                edge                node {$0.4$}  (4)
                edge                node {$0.25$} (6)                
            (2) edge                node {$0.1$}  (1)
                edge [bend left]    node {$0.1$}  (3)
                edge                node {$0.8$}  (4)
            (3) edge                node {$0.5$}  (2)
                edge [bend right]   node {$0.5$}  (4)
            (4) edge                node {$0.33$} (6)
                edge                node {$0.67$} (5)
            (5) edge [loop below]   node {$0.9$}  (5)
                edge                node {$0.1$}  (3)
            (6) edge [loop above]   node { $1$ }  (6);
\end{tikzpicture}
}
\end{subfigure}
\begin{subfigure}[Modified model $M_1^{(2)}$ in step 2.]{%
 \label{small_step2}
        \begin{tikzpicture}[->,>=stealth',shorten >=1pt,auto,node distance=2.05cm, thin, state/.style={circle,draw,font=\sffamily\tiny,initial text=,minimum size=0.75cm, inner sep =0mm}]
    
    \node[initial,state,label=270 : {\tiny $s_{init}$}] (0)              {$\lnot p q\ t$};
    \node[state]         (2) [right of=0] {$\ p\ q \lnot t $};
    \node[state]         (1) [above of=2] {$p\ \lnot qt$};
    \node[state]         (3) [below of=2] {$\lnot p q\ t$};
    \node[state]         (4) [right of=2] {$p\ q\ t$};
    \node[state]         (5) [right of=3] {$\ p \lnot q\lnot t$};
    \node[state]         (6) [right of=1] {$\ p\ q \lnot t\ $};
  
    \path [every node/.style={font=\sffamily\tiny}]  
            (0) edge [bend left]    node {$0.3$}  (1)
                edge                node {$0.2$}  (2)
                edge                node {$0.5$}  (3)            
            (1) edge                node {$0.1$}  (0)
                edge [bend left]    node {$0.35$} (2)
                edge                node {$0.4$}  (4)
                edge                node {$0.25$} (6)                
            (2) edge                node {$0.1$}  (1)
                edge [bend left]    node {$0.1$}  (3)
                edge                node {$0.8$}  (4)
            (3) edge                node {$0.5$}  (2)
                edge [bend right]   node {$0.5$}  (4)
            (4) edge                node {$0.33$} (6)
                edge                node {$0.67$} (5)
            (5) edge [loop below]   node {$0.9$}  (5)
                edge                node {$0.1$}  (3)
            (6) edge [loop above]   node { $1$ }  (6);
\end{tikzpicture}
}\end{subfigure}
\caption{Two step modification of qDTMC $M_1$ in Fig. \ref{example_new}.}
\end{figure}
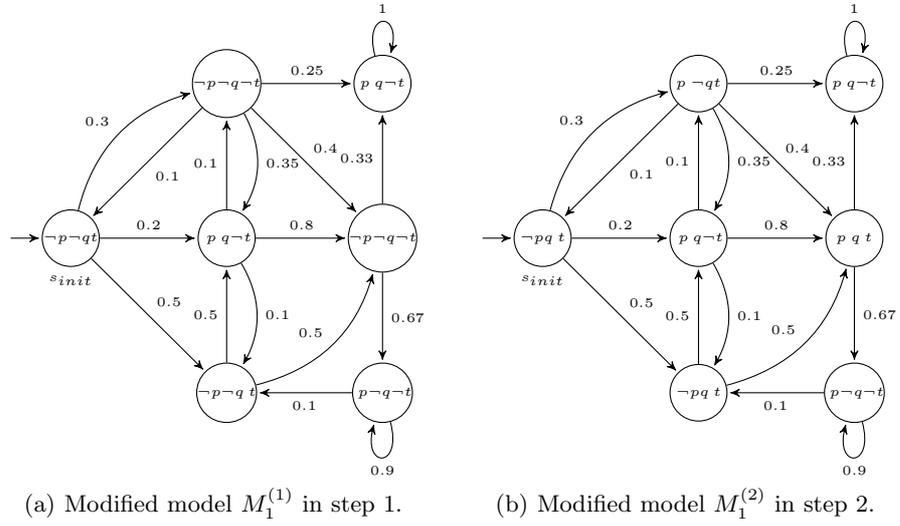

\begin{theorem}
The algorithm qMC solves the model checking problem for qPCTL: for a qDTMC $M$ and a qPCTL formula $\Phi$,
\begin{itemize}
\item qMC($M,\Phi$)=T (alt., F or ?) iff $(s_{init},\Phi)$=T (resp., F or ?)
\end{itemize}
\end{theorem}

\begin{proof}
The algorithm qMC solves the model checking problem for qPCTL, if it matches the semantics of qPCTL for all 
state formulas. Recall that there are two types of state formulas: non-probabilistic and probabilistic. 
The proof that the algorithm works for non-probabilistic state formulas is straightforward and omitted here. 

If 
$\psi$ is a path formula in a probabilistic state formula then:
\[
qMC(M,Pr_{\geq\theta}(\psi))=\textrm{T (alt., F or ?) iff } (s_{init},Pr_{\geq\theta}(\psi))\textrm{=T (resp., F or ?)}
\]
Recall that there are three path formulas : $X\Phi$, $\Phi_1U^{\leq k} \Phi_2$ and $\Phi_1U\Phi_2$. The state formulas 
in these path formulas could in turn also be either non-probabilistic or probabilistic. This allows the algorithm 
to verify properties with both nested and non-nested path formulas. We now prove that the algorithm qMC matches the semantics 
of probabilistic state formulas.

\begin{itemize}
\item  $Pr_{\geq\theta}(X\Phi)$ : 
The correctness of \textit{next} operator can be proved through the following claims:
\begin{claim}
In the second conditional of Algorithm $qMC$, the formula $Pr_{\geq\theta}(X\Phi)$ evaluates to $T$ in the qDTMC $M$ 
if and only if it evaluates to $T'$ in the binary DTMC $M^{(1)}$.
\end{claim}
\begin{proof}
The state formula $Pr_{\geq\theta}(X\Phi)$ evaluates to $T$ in the qDTMC $M$, if there are at least $\theta$ fraction of 
paths that evaluate $X\Phi$ to $T$ in $M$. The mapping of the truth values while constructing the binary DTMC $M^{(1)}$ 
in the second conditional does not disturb $T$. So, there continues to be at least $\theta$ fraction of the paths evaluating $ X\Phi$ to 
$T$ in $M^{(1)}$. Therefore, $Pr_{\geq\theta}(X\Phi)$ evaluates to $T$ in the qDTMC $M$, if and only if it evaluates to $T'$ 
in the binary DTMC $M^{(1)}$.
\end{proof}

\begin{claim}
In the third conditional of Algorithm $qMC$, the formula $Pr_{\geq\theta}(X\Phi)$ evaluates to $F$ in the qDTMC $M$ 
if and only if it evaluates to $F'$ in the binary DTMC $M^{(2)}$.
\end{claim}
\begin{proof}
The state formula $Pr_{\geq\theta}(X\Phi)$ should evaluate to $F$ in the qDTMC $M$, if there are at least $1-\theta$ 
fraction of paths that evaluate $X\Phi$ to $F$ in $M$. The construction of the binary 
DTMC $M^{(2)}$ does not disturb $F$ in the third conditional. If the binary model checker returns $F'$ for 
$Pr_{\geq\theta}(X\Phi)$ in spite of $?$ being identified with $T$, it implies that more than $1-\theta$ fraction of 
the paths evaluate to $F$ in both $M$ and $M^{(2)}$. Thus, $Pr_{\geq\theta}(X\Phi)$ evaluates to $F$ in the qDTMC $M$, 
if and only if it evaluates to $F'$ in the binary DTMC $M^{(2)}$.
\end{proof}

\begin{claim}
In the third conditional of Algorithm $qMC$, the formula $Pr_{\geq\theta}(X\Phi)$ evaluates to $?$ in the qDTMC $M$
if and only if it evaluates to $T'$ in the binary DTMC $M^{(2)}$.
\end{claim}
\begin{proof}
In the third conditional, truth values $T$ and $?$ in the qDTMC $M$ are mapped to $T$ in the binary DTMC $M^{(2)}$. If 
the formula $Pr_{\geq\theta}(X\Phi)$ evaluates to $T'$ in $M^{(2)}$, then it means that at least $\theta$ fraction of paths 
evaluate $X\Phi$ to $T'$ in $M^{(2)}$. This fraction is the sum of fractions of the paths in which 
$X\Phi$ is either $T$ or $?$ in $M$. 

Let the fractions of paths that evaluate $X\Phi$ to $T$ and $F$ in $M$ be $p$ and $q$ respectively. 
Then reaching the third conditional implies that $p < \theta$. 
Since the formula $Pr_{\geq\theta}(X\Phi)$ evaluated to $T'$ in $M^{(2)}$, it is clear that $q < 1-\theta$. Therefore, 
there exists a fraction of paths in $M$ for which $X \Phi$ evaluates to $?$, such that the formula $Pr_{\geq\theta}(X\Phi)$ 
can neither be $T$ nor $F$ in the qDTMC $M$. Hence, the formula is evaluated to $?$ in the qDTMC $M$. Thus, an output 
of $T'$ by the binary model checker is correctly interpreted as $?$ in the qDTMC $M$. 
\end{proof}

\begin{example}
Given a qDTMC $M_{1}$ in figure \ref{example_new}, a property $\phi=\Pr_{\geq\theta}(X\ p)$ 
can be verified for different values of $\theta$. For instance, if $\theta = 0.1$, then binary DTMC $M_1^{(1)}$, 
given in figure \ref{small_step1}, has atleast $\theta$ fraction of the paths that evaluate 
to $T'$. Thus, the property evaluates to $T$ in qDTMC $M_1$.

However, if $\theta = 0.8$, then $M_1^{(1)}$ does not have sufficient fraction of paths evaluating to $T$. The algorithm 
now modifies $M_1$ to $M_{1}^{(2)}$, given in figure \ref{small_step2}, to check if sufficient number of paths that
disprove $\phi$ exist. Since more than $1-\theta$ fraction of paths in $M_{1}^{(2)}$ evaluate to $F'$, 
for $\theta = 0.8$, the property is evaluated to $F$ in qDTMC $M_1$.

Similarly, for $\theta = 0.4$, $M_1^{(1)}$ does not have sufficient fraction of paths evaluating to $T$. But, $M_{1}^{(2)}$ 
also does not have sufficient paths evaluating to $F'$. In such a case, the
property $\Phi$ is evaluated to $?$ in qDTMC $M_1$ due to lack of 
sufficient conclusive paths.

\end{example}

\item $Pr_{\geq\theta}(\Phi_1U^{\leq k} \Phi_2)$ : 
The correctness of \textit{bounded until} operator is similarly proved using following claims:

\begin{claim}
In the second conditional, $Pr_{\geq\theta}(\Phi_1U^{\leq k} \Phi_2)$ evaluates to $T$ in the qDTMC $M$ if and only if 
it evaluates to $T'$ in the binary DTMC $M^{(1)}$.
\end{claim}
\begin{proof}
The formula $Pr_{\geq\theta}(\Phi_1U^{\leq k} \Phi_2)$ is evaluated to $T$ in the qDTMC $M$, if there are at least 
$\theta$ fraction of paths that evaluate $\Phi_1U^{\leq k} \Phi_2$ to $T$ in $M$. Also, the truth value mapping during the 
construction of the binary DTMC $M^{(1)}$ does not alter $T$. Therefore, if $\Phi_1$ holds on a path until $\Phi_2$ becomes 
true in the qDTMC $M$, it will continue to remain that way in the binary DTMC $M^{(1)}$.  
Hence, $Pr_{\geq\theta}(\Phi_1U^{\leq k} \Phi_2)$ is evaluated to $T$ in the qDTMC $M$ if and only if it evaluated to $T'$ in the 
binary DTMC $M^{(1)}$.
\end{proof}

\begin{claim}
In the third conditional, $Pr_{\geq\theta}(\Phi_1U^{\leq k} \Phi_2)$ evaluates to $F$ in the qDTMC $M$ if and only if
it evaluates to $F'$ in the  binary DTMC $M^{(2)}$.
\end{claim}
\begin{proof}
In the qDTMC $M$, a path can evaluate $\Phi_1U^{\leq k}\Phi_2 $  to  $F$ if one of the following happens:
\begin{itemize}
\item $\Phi_2$ is $F$ all along the path up to the $k^{th}$ state. 
\item $\Phi_2$ is not $F$ for some $i\leq k$, but $\Phi_1$ is $F$ for some $j<i$. 
\end{itemize}

Recall that the state formula $Pr_{\geq\theta}(\Phi_1U^{\leq k}\Phi_2 )$ evaluates to $F$ in the qDTMC $M$, if there are 
at least $1-\theta$ fraction of paths that evaluate $\Phi_1U^{\leq k}\Phi_2$ to $F$ in $M$. The mapping of truth values 
does not change $F$ in third conditional. So the paths that evaluated the formula $\Phi_1U^{\leq k}\Phi_2$ to $F$ in 
the qDTMC $M$ would continue to do so in binary DTMC $M^{(2)}$. So, there are at least $1-\theta$ fraction of the 
paths that evaluate $\Phi_1U^{\leq k} \Phi_2$ to $F$ in the qDTMC, if and only if at least $1-\theta$ fraction of the 
paths evaluate the formula to $F'$ in the binary DTMC $M^{(2)}$. Thus, $Pr_{\geq\theta}(\Phi_1U^{\leq k} \Phi_2)$ 
evaluates to $F$ in the qDTMC $M$ if and only if it evaluates to $F'$ in the  binary DTMC $M^{(2)}$.
\end{proof}

\begin{claim}
In the third conditional, the formula $Pr_{\geq\theta}(\Phi_1U^{\leq k} \Phi_2)$ evaluates to $?$ in the qDTMC $M$ if 
and only if it evaluates to $T'$ in the binary DTMC $M^{(2)}$.
\end{claim}
\begin{proof}

Recall that the formula $Pr_{\geq\theta}(\Phi_1U^{\leq k} \Phi_2)$ evaluates to $?$ in qDTMC $M$, if one of 
the following occurs :

\begin{itemize}
\item $\Phi_2$ is $?$ in all the states up to $k$.
\item $\Phi_2$ is $?$ in at least one state and is never $T$ in any of the states along the path up to $k$, 
			but $\Phi_1$ is never $F$. 
\item $\Phi_2$ is $T$ for some $i\leq k$ and $\Phi_1$ is $?$ in at least one state but never $F$ in any of  
			the states up to $i$.

\end{itemize}

For the third conditional, the truth values $T$ and $?$ in the qDTMC $M$ are mapped to $T$ in the binary DTMC $M^{(2)}$. 
So, in all the above cases, the binary model checker outputs $T'$, because $?$ is mapped to $T$ in binary DTMC $M^{(2)}$. 
If the binary model checker returns $T'$ at the third conditional, then there are at least $\theta$ fraction of paths in 
the qDTMC $M$ that evaluate to either $T$ or $?$. If the fraction of paths that evaluate $\Phi_1U^{\leq k} \Phi_2$ to 
$T$ in $M$ is $p$, then from the second conditional, $p < \theta$. Further, let the fraction of paths that evaluate $X\Phi$ 
to $F$ in $M$ be $q$. If the formula $Pr_{\geq\theta}(\Phi_1U^{\leq k} \Phi_2)$ evaluates to $T'$ in $M^{(2)}$, it is 
clear that $q < 1-\theta$. It can then easily be concluded that there do not exist sufficient number of conclusive paths 
(either $T$ or $F$) 
in the qDTMC $M$, and the formula $Pr_{\geq\theta}(\Phi_1U^{\leq k} \Phi_2)$ evaluates to $?$.

\end{proof}

\item $Pr_{\geq\theta}(\Phi_1U^{\leq k} \Phi_2)$ : The proof for \textit{unbounded until} operator is a simple extension 
of the \textit{bounded until} operator. 
\end{itemize}
\qed
\end{proof}

\begin{remark}
If the state formula occurs in negated form as $\lnot\Phi$, then we use $\Phi'=\lnot\Phi$ in the $qMC$
algorithm, proceed as usual and negate the final answer as per the semantics of three valued logic.
\end{remark}

\begin{remark}
If the probabilistic query is of the type $Pr_{<\theta}(\psi)$, we use the identity 
$Pr_{<\theta}(\psi)=\lnot Pr_{\geq \theta} (\psi)$ and proceed as usual.
\end{remark}

\begin{remark}
The Algorithm \ref{qMC} discussed here is symmetric in the sense that the result would not change if the order 
of truth value mapping is swapped in the algorithm. 
\end{remark}

\section{Implementation and Results}

We use PRISM~\cite{prism} for the binary model checker subroutine in the implementation of the $qMC$ 
algorithm. The algorithm works for both numerical and statistical model checking. The inputs to the model 
checker are the three valued probabilistic model and the property specification. The model checker then 
verifies the input property in the given model. If the input property contains nested probabilistic 
operators, then each inner probabilistic formula is considered as a separate property and verified first. 
The results of these sub-formulas are then replaced in the input property to remove nesting. 
However, the current version of PRISM does not support statistical model checking of nested properties.

\subsection{Results with qDTMC}

We illustrate our approach with the qDTMCs $M_1$ (Fig~\ref{example}), $M_2$ (Fig~\ref{small_more}), $M_3$ (Fig~\ref{input_less}) and $M_4$ (Fig~\ref{input_more}).
Note that $M_1$ and $M_2$ 
(and $M_3$ and $M_4$) have the same state space and differ only in the number of \emph{unknowns}.
 These models are checked against different properties to observe the effect of unknown information on the behaviour of the models. 

The first set of verification tests was done on two small qDTMCs, $M_1$ and $M_2$, given in Fig~\ref{example} 
and Fig~\ref{small_more}. In these models, there are two atomic propositions $p$ and $q$, each of which can have 
a truth value from the set $\{T,F, ?\}$. These models are verified for two properties: 
\(\Phi_1\ = \Pr_{\geq \theta}(\lnot p\ U\ r)\) and \(\Phi_2\ = \Pr_{\geq \theta}(X q)\). An additional atomic 
proposition $t \equiv \lnot p $ is added in the models to handle the negation in $\Phi_1$. Thus, $\Phi_1$ can 
now be written as $\Pr_{\geq \theta}(t\ U\ r)$. The results corresponding to various values of $\theta$ in 
$\Phi_1$ and $\Phi_2$, for qDTMCs $M_1$ and $M_2$, are in tables \ref{small_prop1} and \ref{small_prop2}, respectively.

\begin{figure}
\centering

        \begin{tikzpicture}[->,>=stealth',shorten >=1pt,auto,node distance=2.05cm, thin, 
        state/.style={circle,draw,font=\sffamily\tiny,initial text=,minimum size=0.75cm, inner sep =0mm}]
        \node[initial,state,label=270 : {\tiny $s_{init}$}] (0)              {$ \lnot p q? t?$};
    \node[state]         (2) [right of=0] {$\ p\ q \lnot t $};
    \node[state]         (1) [above of=2] {$p? q?t?$};
    \node[state]         (3) [below of=2] {$ p? q?t?$};
    \node[state]         (4) [right of=2] {$p?q?t?$};
    \node[state]         (5) [right of=3] {$\ p \lnot q\lnot t$};
    \node[state]         (6) [right of=1] {$\ p\ q\ \lnot t$};
  
    \path [every node/.style={font=\sffamily\tiny}]  
            (0) edge [bend left]    node {$0.3$}  (1)
                edge                node {$0.2$}  (2)
                edge                node {$0.5$}  (3)            
            (1) edge                node {$0.1$}  (0)
                edge [bend left]    node {$0.35$} (2)
                edge                node {$0.4$}  (4)
                edge                node {$0.25$} (6)                
            (2) edge                node {$0.1$}  (1)
                edge [bend left]    node {$0.1$}  (3)
                edge                node {$0.8$}  (4)
            (3) edge                node {$0.5$}  (2)
                edge [bend right]   node {$0.5$}  (4)
            (4) edge                node {$0.33$} (6)
                edge                node {$0.67$} (5)
            (5) edge [loop below]   node {$0.9$}  (5)
                edge                node {$0.1$}  (3)
            (6) edge [loop above]   node { $1$ }  (6);
\end{tikzpicture}
\caption{Another example qDTMC with small state-space. This qDTMC has same state space as $M_1$, but more number of unknowns.}
\label{small_more}
\end{figure}
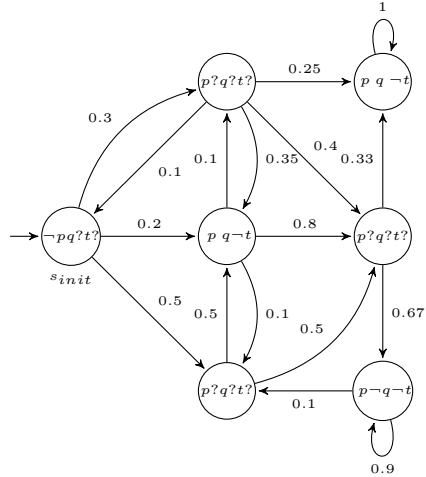

\begin{table}
\centering
    \begin{tabular}{|c|c|c|c|c|c|c|c|c|c|}
          \hline
          &&&&&&&&& \\
          $\theta$ &  0.1  &  0.2  &  0.3  &  0.4  &  0.5  &  0.6  &  0.7  &  0.8  &  0.9  \\[6pt]
         \hline &&&&&&&&& \\
          $M_1$  & T & T & T & T & ? & ? & ? & ? & ? \\
           \hline &&&&&&&&& \\
          $M_2$  & T & T & ? & ? & ? & ? & ? & ? & ? \\ 
          \hline
    \end{tabular}
    \caption{Results for various values of $\theta$ for the property \(\Phi_1\ = \Pr_{\geq \theta}(\lnot p\ U\ r)\)}
    \label{small_prop1}
\centering
    \begin{tabular}{|c|c|c|c|c|c|c|c|c|c|}
          \hline
          &&&&&&&&& \\
          $\theta$ &  0.1\ &  0.2 &  0.3 &  0.4 &  0.5 &  0.6 &  0.7 &  0.8 &  0.9  \\[6pt]
         \hline &&&&&&&&& \\
          $M_1$   & T & T & ? & ? & ? & ? & ? & F & F \\
           \hline &&&&&&&&& \\
          $M_2$ & T & T & ? & ? & ? & ? & ? & ? & ? \\ 
          \hline
    \end{tabular}
    \caption{Results for various values of $\theta$ for the property $\Phi_2 = \Pr_{\geq \theta}(Xq)$}
    \label{small_prop2}
\end{table}

It is evident from the results that when a model checker does not have enough paths to either accept or reject a proposition, 
it generates $?$ as result. The $?$ truth value indicates to the model designer that the model needs more details, so as to be 
certain of its required behavior. On the other hand, a $T$ or an $F$ result concludes that in spite of missing information, 
the property or the behavior of the model can be easily verified. The difference in results for models $M_1$ and $M_2$ also 
shows that with the increase in absence or uncertainty of the information in the models, uncertainty in the behavior 
of the model increases as well.

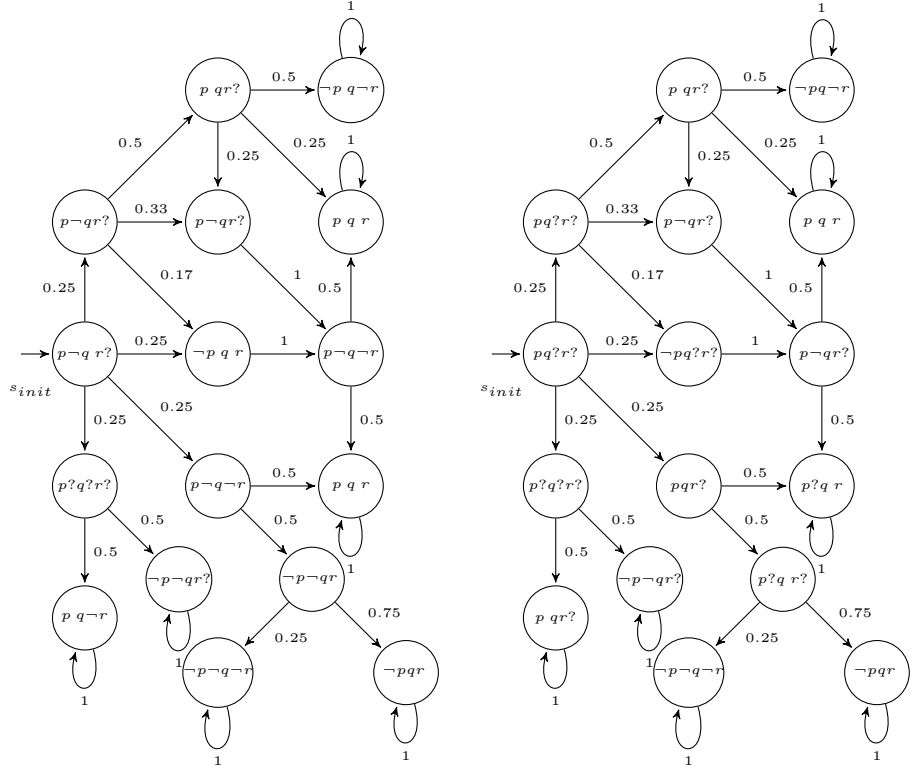
\begin{figure}
\centering
    \begin{subfigure}[qDTMC $M_3$ with less unknown values]{%
     \label{input_less}
        \begin{tikzpicture}[->,>=stealth',shorten >=1pt,auto,node distance=1.75cm, thin, state/.style={circle,draw,font=\sffamily\tiny,initial text=,minimum size=0.85cm, inner sep =0mm}]
    \node[initial,state,label=225 : {\tiny $s_{init}$}] (0)              {$ p \lnot q\ r?$ };
    \node[state]         (1)  [above of=0] {$p \lnot q r?$};
    \node[state]         (2)  [right of=0] {$\lnot p\ q \ r$};
    \node[state]         (3)  [below of=2] {$ p \lnot q \lnot r$};
    \node[state]         (4)  [below of=0] {$p?q?r?$};
    \node[state]         (6)  [right of=1] {$\ p \lnot q r?$};
    \node[state]         (5)  [above of=6] {$\ p \ qr?$};
    \node[state]         (12) [right of=2] {$p \lnot q \lnot r$};
    \node[state]         (7)  [below of= 12] {$p \ q\  r$};
    \node[state]         (8) [below right of=3] {$\lnot p \lnot q r\ $};
    \node[state]         (9) [below right of=4] {$\lnot p \lnot q r?$};
    \node[state]         (10) [right of=5] {$\lnot p \ q \lnot r\ $};
    \node[state]         (11) [right of=6] {$ p \ q \ r $};
    \node[state]         (13) [below right of=8] {$\lnot pqr\ $};
    \node[state]         (14) [below left of=8] {$\lnot p\lnot q\lnot r$};
    \node[state]         (15) [below of=4] {$p \ q \lnot r$};
    \path [every node/.style={font=\sffamily\tiny}]  
            (0) edge 			    node {$0.25$}  (1)
                edge                node {$0.25$}  (2)
                edge           node {$0.25$}  (3)
                edge                node {$0.25$}  (4)
            (1) edge                node {$0.5$}   (5)
                edge 			    node {$0.17$}  (2)
                edge                node {$0.33$}  (6)
            (2) edge                node {$1$}	   (12)
            (3) edge                node {$0.5$}   (7)
                edge   				node {$0.5$}   (8)
            (4) edge                node {$0.5$}   (15)
                edge                node {$0.5$}   (9)
            (5) edge			    node {$0.5$}   (10)
                edge                node {$0.25$}  (11)
          		edge                node {$0.25$}  (6)
            (6) edge 				node { $1$ }   (12)
            (7) edge [loop below] 	node { $1$ }   (7)
            (8) edge 				node {$0.75$}  (13)
                edge 				node {$0.25$}  (14)
            (9) edge [loop below] 	node { $1$ }   (9)
            (10) edge [loop above] 	node { $1$ }   (10)
            (11) edge [loop above] 	node { $1$ }   (11)
            (12) edge 				node {$0.5$}   (11)
                 edge 				node {$0.5$}   (7)
            (13) edge [loop below] 	node { $1$ }   (13)
            (14) edge [loop below] 	node { $1$ }   (14)
            (15) edge [loop below] 	node { $1$ }   (15);
\end{tikzpicture}
}
\end{subfigure}
    \begin{subfigure}[qDTMC $M_4$ with more unknown values]{%
     \label{input_more}
        \begin{tikzpicture}[->,>=stealth',shorten >=1pt,auto,node distance=1.75cm, thin, state/.style={circle,draw,font=\sffamily\tiny,initial text=,minimum size=0.85cm, inner sep =0mm}]
    \node[initial,state,label=225 : {\tiny $s_{init}$}] (0)              {$ p q? r?$};
    \node[state]         (1)  [above of=0] {$p q? r?$};
    \node[state]         (2)  [right of=0] {$\lnot p q? r?$};
    \node[state]         (3)  [below of=2] {$ p q r?$};
    \node[state]         (4)  [below of=0] {$p?q?r?$};
    \node[state]         (6)  [right of=1] {$p \lnot q r?$};
    \node[state]         (5)  [above of=6] {$ p \ qr?$};
    \node[state]         (12) [right of=2] {$\ p \lnot q r? $};
    \node[state]         (7)  [below of= 12] {$p? q\  r$};
    \node[state]         (8) [below right of=3] {$p? q\ r? $};
    \node[state]         (9) [below right of=4] {$\lnot p \lnot q r?$};
    \node[state]         (10) [right of=5] {$\lnot p q \lnot r $};
    \node[state]         (11) [right of=6] {$ p \ q \ r$};
    \node[state]         (13) [below right of=8] {$\lnot pqr\ $};
    \node[state]         (14) [below left of=8] {$\lnot p\lnot q\lnot r$};
    \node[state]         (15) [below of=4] {$p \ q  r? $};
    \path [every node/.style={font=\sffamily\tiny}]  
            (0) edge 			    node {$0.25$}  (1)
                edge                node {$0.25$}  (2)
                edge           node {$0.25$}  (3)
                edge                node {$0.25$}  (4)
            (1) edge                node {$0.5$}   (5)
                edge 			    node {$0.17$}  (2)
                edge                node {$0.33$}  (6)
            (2) edge                node {$1$}	   (12)
            (3) edge                node {$0.5$}   (7)
                edge   				node {$0.5$}   (8)
            (4) edge                node {$0.5$}   (15)
                edge                node {$0.5$}   (9)
            (5) edge			    node {$0.5$}   (10)
                edge                node {$0.25$}  (11)
          		edge                node {$0.25$}  (6)
            (6) edge 				node { $1$ }   (12)
            (7) edge [loop below] 	node { $1$ }   (7)
            (8) edge 				node {$0.75$}  (13)
                edge 				node {$0.25$}  (14)
            (9) edge [loop below] 	node { $1$ }   (9)
            (10) edge [loop above] 	node { $1$ }   (10)
            (11) edge [loop above] 	node { $1$ }   (11)
            (12) edge 				node {$0.5$}   (11)
                 edge 				node {$0.5$}   (7)
            (13) edge [loop below] 	node { $1$ }   (13)
            (14) edge [loop below] 	node { $1$ }   (14)
            (15) edge [loop below] 	node { $1$ }   (15);
\end{tikzpicture} }
\end{subfigure}
\caption{Example qDTMCs with large state-space}
\end{figure}

The qDTMCs $M_3$ and $M_4$ have a larger state-space as can be see in Figures \ref{input_less} and \ref{input_more}. Similar to the 
previous case, these models also differ only in labeling functions but contain three atomic propositions 
$p$, $q$ and $r$. These models are verified for both non-nested and nested properties. The behaviour of qDTMCs 
$M_3$ and $M_4$ are verified using properties $\Phi_2 = \Pr_{\geq \theta}(Xq)$, $\Phi_3 = \Pr_{\geq \theta}(p\ U\ r)$, 
$\Phi_4 = \Pr_{\geq \theta}(p\ U\ \Pr_{\geq 0.8}(Xr))$ and $\Phi_5 = \Pr_{\geq \theta}(\Pr_{\geq 0.2}(p\ U\  r)\  U\ q)$. 
The results for these properties, for various values of $\theta$ can be found in Tables \ref{large_prop1}, 
\ref{large_prop2}, \ref{large_prop3} and \ref{large_prop4}. The results for these models concur with the ones 
for models $M_1$ and $M_2$, and same conclusions can be made in this case as well.

\begin{table}
\centering
    \begin{tabular}{|c|c|c|c|c|c|c|c|c|c|}
          \hline
          &&&&&&&&& \\
          $\theta$ &  0.1\ &  0.2 &  0.3 &  0.4 &  0.5 &  0.6 &  0.7 &  0.8 &  0.9  \\[6pt]
         \hline &&&&&&&&& \\
          $M_3$   & T & T & ? & ? & ? & F & F & F & F \\
           \hline &&&&&&&&& \\
          $M_4$ & T & T & ? & ? & ? & ? & ? & ? & ? \\ 
          \hline
    \end{tabular}
    \caption{Results for various values of $\theta$ for the property $\Phi_2 = \Pr_{\geq \theta}(Xq)$}, 
    \label{large_prop1}
\centering
    \begin{tabular}{|c|c|c|c|c|c|c|c|c|c|}
          \hline
          &&&&&&&&& \\
          $\theta$ &  0.1\ &  0.2 &  0.3 &  0.4 &  0.5 &  0.6 &  0.7 &  0.8 &  0.9  \\[6pt]
         \hline &&&&&&&&& \\
          $M_3$   & T & T & T & T & T & T & ? & ? & ? \\
           \hline &&&&&&&&& \\
          $M_4$ & T & T & ? & ? & ? & ? & ? & ? & ? \\ 
          \hline
    \end{tabular}
    \caption{Results for various values of $\theta$ for the property $\Phi_3 = \Pr_{\geq \theta}(p\ U\ r)$}
    \label{large_prop2}
\centering
    \begin{tabular}{|c|c|c|c|c|c|c|c|c|c|}
          \hline
          &&&&&&&&& \\
          $\theta$ &  0.1\ &  0.2 &  0.3 &  0.4 &  0.5 &  0.6 &  0.7 &  0.8 &  0.9  \\[6pt]
         \hline &&&&&&&&& \\
          $M_3$   & T & T & T & ? & ? & ? & F & F & F \\
           \hline &&&&&&&&& \\
          $M_4$ & T & T & ? & ? & ? & ? & ? & ? & ? \\ 
          \hline
    \end{tabular}
    \caption{Results for various values of $\theta$ for the property $\Phi_4 = \Pr_{\geq \theta}(p\ U\ \Pr_{\geq 0.8}(Xr))$}
    \label{large_prop3}

\centering
    \begin{tabular}{|c|c|c|c|c|c|c|c|c|c|}
          \hline
          &&&&&&&&& \\
          $\theta$ &  0.1\ &  0.2 &  0.3 &  0.4 &  0.5 &  0.6 &  0.7 &  0.8 &  0.9  \\[6pt]
         \hline &&&&&&&&& \\
          $M_3$   & T & T & T & T & T & T & T & ? & ? \\
           \hline &&&&&&&&& \\
          $M_4$ & T & T & T & T & ? & ? & ? & ? & ? \\ 
          \hline
    \end{tabular}
    \caption{Results for various values of $\theta$ for the property $\Phi_5 = \Pr_{\geq \theta}(\Pr_{\geq 0.2}(p\ U\  r)\  U\ q)$}
    \label{large_prop4}
\end{table}

\subsection{\emph{Unknowns} in Code}
Finally, we present an example of a code listing that has a module whose implementation details (and hence correctness properties etc) are
\emph{unknown}. 
This could be due to several reasons: for instance, if a module in the system is not implemented yet and exists only in ``stub" form
or if an implementation of the module exists, but whose correctness is not established. Therefore, it would be 
incorrect to assume any truth value for certain atomic propositions in the states that represent such a module. 

Listing \ref{snippet} shows a code snippet of a system wherein $func$ is a function call whose internal
working is not known to the system designer. The value thus returned by this function is not known. 
The system contains three atomic propositions, $p : var1=10$, $q : var2=z$ and $r : var3 \geq 0$, which 
are true if and only if their respective conditions hold true in the system.  
The truth values of these atomic propositions change with each set of assignment statements in the code. The module can now be 
modeled as a $qDTMC$ using the above atomic propositions, as shown in Fig. \ref{code}. 
Each state in the qDTMC represents the possible truth values of atomic propositions during a code execution. For instance, when the variables 
$var1$, $var2$ and $var3$ are initialized to -1, then all three atomic propositions are false in the initial state of the qDTMC.

Algorithm  qMC can now evaluate various properties for this module. The results of an example qPCTL query 
$\Pr_{\geq \theta}(\lnot q \ U \ p)$ are shown in Table~\ref{code_DTMC}.

\begin{lstlisting}[caption={Code snippet},frame=single,label=snippet,captionpos=b,language= C,basicstyle=\small]
int x= randint(1,5);
int y= randint(1,10);
\\randint(int a,int b):returns a random int between a and b

int z=10;
int var1=-1, var2=-1, var3=-1;

var1= x+y;
if (var1%2==0)
    var2=z;
else
    var2=func(z);
if (var1%3 == 0) {
    var1=5;
    var2=7;
    var3=0;	  }
else {
    var1=10;
    var3=func(z);  }

\end{lstlisting}

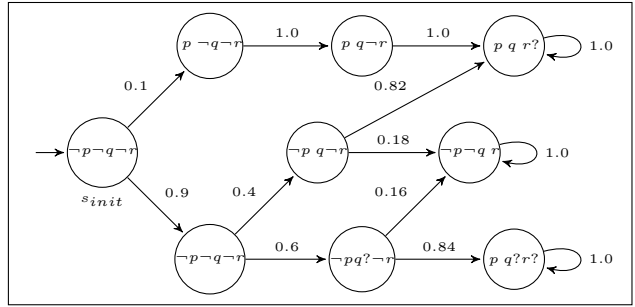
\begin{figure}[!ht]{}
\centering
\begin{tikzpicture}[framed,->,>=stealth',shorten >=1pt,auto,node distance=2cm, thin, state/.style={circle,draw,font=\sffamily\tiny,initial text=,minimum size=0.8cm, inner sep =0mm}]
    
    \node[initial,state,label=270 : {\tiny $s_{init}$}] (0)              {$\lnot p \lnot q \lnot r$};
    \node[state]         (2) [below right of=0] {$\lnot p \lnot q \lnot r$};
    \node[state]         (1) [above right of=0] {$\ p\ \lnot q \lnot r$};
    \node[state]         (6) [right of=2] {$\lnot p q? \lnot r$};
    \node[state]         (3) [right of=1] {$\ p\ q \lnot r$};
    \node[state]         (4) [right of=3] {$\ p \ q\ r?$};
    \node[state]         (5) [above right of=2] {$\lnot p\ q \lnot r $};
    \node[state]         (7) [right of=5] {$\lnot p \lnot q \ r$};
    \node[state]         (10)[right of=6] {$p\ q? r?$};

    \path [every node/.style={font=\sffamily\tiny}]  
            (0) edge 			    node {$0.1$}  (1)
                edge                node {$0.9$}  (2)
            (1) edge                node {$1.0$}  (3)
            (2) edge                node {$0.4$}  (5)
                edge 			    node {$0.6$}  (6)
            (3) edge                node {$1.0$}  (4)
            (4) edge [loop right]   node {$1.0$}  (4)
            (5) edge 		   		node {$0.18$}  (7)
                edge                node {$0.82$}  (4)
            (6) edge 		   node { $0.16$ }  (7)
            	edge 		   node { $0.84$ }  (10)
            (7) edge [loop right]   node {$1.0$}  (7)
            (10) edge [loop right]	node {$1.0$}  (10) ;
\end{tikzpicture}
\caption{qDTMC $M_5$ for code snippet in Listing \ref{snippet}}
\label{code}
\end{figure}

 \begin{table}
\centering
\begin{tabular}{|c|c|c|c|c|c|c|c|c|c|}
          \hline
          &&&&&&&&& \\
          $\theta$ &  0.1\ &  0.2 &  0.3 &  0.4 &  0.5 &  0.6 &  0.7 &  0.8 &  0.9  \\[6pt]
         \hline &&&&&&&&& \\
          $M_5$   & T & ? &  ? & ? & ? & F & F & F & F \\
          \hline
    \end{tabular}
\caption{Results for $ \Pr_{\geq \theta} (\lnot q U p)$}
\label{code_DTMC}
\end{table}    

\section{Conclusions and Future Directions}
In this paper we presented a technique to determine the feasibility of model checking in the
presence of uncertainty in the implementation of a model of a stochastic system.

We have presented a small example of a code with an incompletely determined module.
However, it involved a manual conversion of the code to a qDTMC for the purpose of model checking.
It remains to be seen how well this conversion and our technique scales to a full-fledged
code of, say, a discrete event simulator.

It would also be interesting to see how this approach can be applied to other modeling-formalism/query-logic 
pairs.

\bibliographystyle{splncs03}
\bibliography{bib}

\end{document}